\documentclass[envcountsame]{llncs}
\usepackage{amsmath, amssymb}
\usepackage[utf8]{inputenc}
\usepackage[obeyFinal]{todonotes}
\usepackage{braket}
\usepackage{booktabs}
\usepackage{tabularx}
\newcolumntype{Y}{>{\centering\arraybackslash}X}
\usepackage[numbers,sort]{natbib}
\usepackage{usebib}
\usepackage{natbib}
 \usepackage{hyperref}%
\usepackage{wrapfig}

\makeatletter
\renewcommand\bibsection%
{
	\section*{\refname
		\@mkboth{\MakeUppercase{\refname}}{\MakeUppercase{\refname}}}
}
\makeatother

\def\N{\mathbb{N}}

\def\sz{\mathsf{0}}
\def\so{\mathsf{1}}
\def\sx{\mathsf{x}}
\def\to{\mathtt{1}}
\DeclareMathOperator{\Fact}{Fact}
\DeclareMathOperator{\Pref}{Pref}
\DeclareMathOperator{\Suff}{Suff}

\DeclareMathOperator{\npal}{npal}
\DeclareMathOperator{\NPal}{NPal}
\DeclareMathOperator{\Pal}{Pal}

\DeclareMathOperator{\word}{\mathsf{serialise}}
\renewcommand{\epsilon}{\varepsilon}

\def\nth#1{#1$^{\text{th}}$}
\def\pnEquiv#1{\equiv_{#1}}

\def\noo#1{|{#1}|_{\so}}

\title{On Collapsing Prefix Normal Words}
\author{Pamela Fleischmann\inst{1} \and Mitja Kulczynski\inst{1} \and Dirk Nowotka\inst{1} \and Danny~B\o gsted~Poulsen\inst{2}}
\authorrunning{P. Fleischmann \and M. Kulczynski \and D. Nowotka \and D. Poulsen}

\institute{
Department of Computer Science, Kiel University,
Kiel, Germany\\
\email{\{fpa,mku,dn\}@informatik.uni-kiel.de}\\~\\
\and
Department of Computer Science, Aalborg University,
Aalborg, Denmark\\
\email{dannybpoulsen@cs.aau.dk}}

\newcommand{\lr}{LR}
\newcommand{\pnPal}{pnPal}

\begin{document}

\maketitle
\begin{abstract}
Prefix normal words are binary words in which each prefix has at least the same 
number of $\so$s  as any factor of the same length. Firstly introduced in 2011, 
the problem of determining the index (amount of equivalence classes for a given 
word length) of the prefix normal equivalence relation is still open. In this 
paper, we investigate two aspects of the problem, namely prefix normal 
palindromes and so-called collapsing words (extending the notion of critical 
words). We prove characterizations for both the palindromes and the collapsing 
words and show their connection. Based on this, we show that still open problems 
regarding prefix normal words can be split into certain subproblems.

\end{abstract}

\section{Introduction}\label{intro}

Two words are called abelian equivalent if the amount of each
letter is identical in both words, e.g.  {\em rotor} 
and {\em torro} are abelian equivalent albeit {\em banana} and 
{\em ananas} are not. Abelian equivalence has been studied with
various generalisations and specifications  such as 
abelian-complexity, $k$-abelian equivalence, avoidability of 
($k$-)abelian powers and much more (cf. e.g.,
\cite{cassaigne2011,coven1973,EhlersMMN15,currie2009recurrent,keranen1992abelian,puzynina2013abelian,richomme2010balance,richomme2010abelian} 
). The number of occurrences of each letter is captured in the
Parikh vector (also known as Parikh image or Parikh mapping)
(\cite{Parikh66}): given a lexicographical order on the alphabet, the \nth{$i$} component of this vector is the amount of the \nth{$i$} letter of the alphabet in a given word. Parikh vectors have been 
studied in \cite{Dassow:2001:PMI,journals/iandc/Karhumaki80,oai:CiteSeerPSU:440186} and are generalised to Parikh matrices for saving more information about the word than just the amount of letters (cf. eg., \cite{Mat04,Salomaa05}).

A recent generalisation of abelian equivalence, for words over the binary 
alphabet $\{\sz,\so\}$,  is  prefix normal 
equivalence (pn-equivalence)~\cite{journals/corr/abs-1805-12405}. Two binary 
words are pn-equivalent if their maximal numbers of $\so$s in any factor of 
length $n$ are 
equal for all $n\in\N$. \citet{journals/tcs/BurcsiFLRS17} showed that this 
relation is indeed an equivalence 
relation and moreover that each class contains exactly one uniquely determined 
representative - called a \emph{prefix normal word}. A word $w$  is said to be prefix normal if the prefix of $w$ of any length has at least the 
number of $\so$s as any of $w$'s factors of the same length. For instance, the 
word $\so\so\sz\so\sz\so$ is prefix normal but $\so\sz\so\so\sz\so$ is not, 
witnessed by the fact that 
$\so\so$ is a factor but not a prefix. Both words are pn-equivalent. In addition 
to being representatives of 
the pne-classes, prefix normal words are also of interest since they are
connected to Lyndon words, in the sense that every prefix normal word
is a pre-necklace \cite{journals/corr/abs-1805-12405}. Furthermore, as shown in \cite{journals/corr/abs-1805-12405}, the indexed jumbled pattern matching problem (see e.g. \cite{journals/ijfcs/BurcsiCFL12,journals/corr/BurcsiFLRS14a,conf/spire/LeeLZ12}) is connected to prefix normal forms: if the prefix normal forms are 
given, the indexed jumbled pattern matching problem can be solved in linear time 
$\mathcal{O}(n)$ of the word length $n$. The best known algorithm for this problem has a run-time of $\mathcal{O}(n^{1.864})$(see \cite{journals/corr/ChanL15}). 
Consequently there is also an interest in prefix normal forms from an algorithmic point of view. An algorithm for the computation of all prefix normal words of length $n$ in run-time $\mathcal{O}(n)$ per word is given in \cite{conf/lata/CicaleseLR18}. \citet{balister2019asymptotic} showed that the number of prefix normal words of length $n$ is $2^{n-\Theta(\log^2(n))}$ and the class of a given prefix normal word contains at most $2^{n-O(\sqrt{n\log(n)})}$ elements. 
A closed formula for the number of prefix normal words is still unknown.
In ``OEIS'' \cite{sloane2003line} the number of prefix normal words of length $n$ (A194850), a list of binary prefix normal words (A238109), and the maximum size of a class of binary words of length $n$ having the same prefix normal form (A238110), can be found. An extension to infinite words is presented in \cite{DBLP:conf/sofsem/CicaleseL019}.

\noindent
 {\bf Our contribution.} 
In this work we investigate two conspicuities mentioned in
 \cite{journals/corr/abs-1805-12405,journals/corr/BurcsiFLRS14}:
 palindromes and extension-critical words. Generalising the result of
 \cite{journals/corr/BurcsiFLRS14} we prove that prefix normal
 palindromes (\pnPal{}) play a special role since they are not pn-equivalent
 to any other word. Since not all palindromes are prefix normal,
 as witnessed by $\so\sz\so\so\sz\so$, determining the number of
 \pnPal s is an (unsolved) sub-problem. We show that
 solving this sub-problem brings us closer to determining the index, i.e. 
number of equivalence classes w.r.t. a  given word length,  of
 the pn-equivalence relation. Moreover we give a characterisation
 based on the maximum-ones function for \pnPal s. The
 notion of extension-critical words is based on an iterative approach:
 compute the prefix normal words of length $n+1$ based on the prefix
 normal words of length $n$. A prefix normal word $w$ is called
 extension-critical if $w\so$ is not prefix normal. For instance, the
 word $\so\sz\so$ is prefix normal but $\so\sz\so\so$ is not and thus
 $\so\sz\so$ is called extension-critical. This means that all
 non-extension-critical words contribute to the class of prefix normal
 words of the next word-length. We investigate the set of
 extension-critical words by introducing an equivalence relation {\em
 collapse}, grouping all extensional-critical words that are
 pn-equivalent w.r.t. length $n+1$. Finally we prove that (prefix
 normal) palindromes and the collapsing relation (extensional-critical
 words) are related.
 In contrast to \cite{journals/corr/abs-1805-12405} we work with suffix-normal words
 (least representatives) instead of prefix-normal words. It follows from Lemma \ref{leastsuffix} that
 both notions lead to the same results.

\noindent
 {\bf Structure of the paper.} In Section 2, the basic definitions and notions 
are presented. In Section 3, we present the results on \pnPal s. Finally, in Section 4, the iterative approach based on collapsing 
words is shown. This includes a lower bound and an upper bound for the number of 
prefix normal words, based on \pnPal s and the collapsing 
relation. Due to space restrictions all proofs are in the appendix.

 \medskip
%
%

\section{Preliminaries}

Let $\N$ denote the set of natural numbers starting with $1$, and let 
$\N_0=\N\cup\{0\}$. Define $[n]=\{1,\dots,n\}$, for $n\in\N$, and set 
$[n]_0=[n]\cup\{0\}$.

An alphabet is a finite set $\Sigma$, the set of all finite words over $\Sigma$ 
is denoted by $\Sigma^{\ast}$, and the empty word by $\varepsilon$. Let 
$\Sigma^+=\Sigma^{\ast}\backslash\{\varepsilon\}$ be the free semigroup for 
the free monoid $\Sigma^{\ast}$. Let $w[i]$ denote the \nth{$i$} letter of $w\in\Sigma^{\ast}$ that is $w=\varepsilon$ or $w=w[1]\dots w[n]$. The {\em length} of a word $w=w[1]\dots w[n]$ is denoted by $|w|$ and let $|\varepsilon|=0$. Set $w[i..j]=w[i]\dots w[j]$ for $1\leq i\leq j\leq |w|$.
Set $\Sigma^n=\{w\in\Sigma^{\ast}|\,|w|=n\}$ 
for all $n\in\N_0$. The number of occurrences of a letter $\sx\in\Sigma$ in 
$w\in\Sigma^{\ast}$ is denoted by $|w|_{\sx}$. For a given word $w\in\Sigma^n$ the 
{\em reversal} of $w$ is defined by $w^R=w[n]\dots w[1]$. A word 
$u\in\Sigma^{\ast}$ is a factor of $w\in\Sigma^{\ast}$ if $w=xuy$ holds for 
some words $x,y\in\Sigma^{\ast}$. If $x=\varepsilon$ then $u$ is called a 
{\em prefix} of $w$ and a {\em suffix} if $y=\varepsilon$. 
Let $\Fact(w), \Pref(w),\Suff(w)$ denote the sets of all factors, prefixes, and suffixes respectively. Define
$\Fact_k(w)=\Fact(w)\cap\Sigma^k$ and $\Pref_k(w)$, $\Suff_k(w)$ are defined accordingly. Notice that 
$|\Pref_k(w)|=|\Suff_k(w)|=1$ for all $k\leq |w|$.
The powers of 
$w\in\Sigma^{\ast}$ are recursively defined by $w^0=\varepsilon$, 
$w^n=ww^{n-1}$ for $n\in\N$.

Following \cite{journals/corr/abs-1805-12405}, we only consider binary alphabets, namely 
$\Sigma=\{\sz,\so\}$ with the fixed lexicographic order induced by $\sz < \so$ on $\Sigma$. In analogy to binary numbers we call a word $w\in\Sigma^n$ {\em odd} if 
$w[n]=\so$ and {\em even} otherwise.

For a function $f:[n]\rightarrow\Delta$ for $n\in\N_0$ and an arbitrary 
alphabet $\Delta$ the concatenation 
of the images defines a finite
word $\word(f)=f(1)f(2)\dots 
f(n)\in\Delta^{\ast}$. Since $\word$  is bijective, we will identify 
$\word(f)$ with 
$f$ and use in both cases $f$ (as long as it is clear from the context).  This 
definition allows 
us to access $f$'s {\em reversed function} $g:[n]\rightarrow\Delta;k\mapsto 
f(n-k+1)$ easily by 
$f^R$. 

\begin{definition}
The {\em maximum-ones functions} is defined for a  word $w\in\Sigma^{\ast}$ by
$f_w:[|w|]_0 \rightarrow [|w|]_0;\,k\mapsto
\max\left\{\,\noo{v} \mid v\in\Fact_k(w)\right\},$ 
giving for each $k\in[|w|]_0$ the maximal number of $\so$s occuring in a factor 
of length $k$. Likewise the {\em prefix-ones and suffix-ones functions} are defined by
$p_w:[|w|]_0 \rightarrow [|w|]_0; k\mapsto |\Pref_k(w)|_{\so}$ and 
$s_w:[|w|]_0 \rightarrow [|w|]_0; k\mapsto |\Suff_k(w)|_{\so}$.
\end{definition}

\begin{definition}
Two words  $u,v\in\Sigma^{n}$ are called {\em prefix-normal equivalent} 
(pn-equivalent, $u\pnEquiv{n}v$) if $f_u=f_v$ holds and $v$'s equivalence class 
is denoted by 
$[v]_{\equiv}=\{u\in\Sigma^{n}|\,u\pnEquiv{n}v\}$. A word $w\in\Sigma^{\ast}$ is 
called {\em prefix (suffix) normal} iff $f_w=p_w$ ($f_w=s_w$ resp.) holds.  
Let $\sigma(w)=\sum_{i\in[n]}f_w(i)$ denote the {\em maximal-one sum}
of a $w\in\Sigma^{n}$.
\end{definition}

\begin{remark}\label{sufpref}
Notice that $s_w=p_{w^R}$, $f_w=f_{w^R}$, $p_w(i),s_w(i)\leq f_w(i)$ for 
all $i\in\N_0$. By $p_{w^R}=s_w$ and 
$f_w=f_{w^R}$ follows immediately that a word $w\in\Sigma^{\ast}$ is prefix 
normal iff its reversal is suffix normal.
\end{remark}
\citet{journals/corr/abs-1805-12405} showed that for  each word
$w\in\Sigma^{\ast}$ there exists exactly one $w'\in [w]_{\equiv}$ that is
prefix normal - the prefix normal form of $w$. We introduce the concept of 
\emph{least representative}, which is the lexicographically smallest element of 
a class and thus also  unique.
As mentioned in \cite{journals/tcs/BurcsiFLRS17} palindromes play a special role.
Immediately by $w=w^R$ for $w\in\Sigma^{\ast}$,  we have $p_w=s_w$, 
i.e. palindromes are 
the only words that can be prefix and suffix normal. Recall that not all 
palindromes are 
prefix normal witnessed by $\so\sz\so\so\sz\so$. 

\begin{definition}
A palindrome is called {\em prefix normal palindrome} (\pnPal{}) if it is prefix
normal. Let $\NPal(n)$ denote the set of all prefix normal palindromes
of length $n\in\N$ 
and set $\npal(n)=|\NPal(n)|$. Let $\Pal(n)$ be the set of all
palindromes of length $n\in\N$.  
\end{definition}

 \begin{table}[t]
 \begin{center}
	\begin{tabular}{c c c}
		\scriptsize word length & \scriptsize prefix normal palindromes & \scriptsize \# prefix normal words\\
		\cmidrule(lr){1-1} \cmidrule(l){2-2} \cmidrule(l){3-3} 
		$1$ & $\sz$, $\so$ & $2$\\
		$2$ & $\sz^2$, $\so^2$ & $3$\\
		$3$ & $\sz^3$, $\so\sz\so, \so^3$& $5$ \\
		$4$ & $\sz^4$, $\so\sz\sz\so, \so^4$& $8$\\
		$5$ & $\sz^5$, $\so\sz\sz\sz\so, \so\sz\so\sz\so, \so\so\sz\so\so, 
\so^5$& $14$\\
		$6$ & $\sz^6$, $\so\sz\sz\sz\sz\so, \so\so\sz\sz\so\so, \so^6$ & $23$\\
		$7$ & $\sz^7$, $\so\sz^2\so\sz^2\so,\so\sz^5\so,\so\sz\so\sz\so\sz\so,\so^2\sz\so\sz\so^2,\so^2\sz^3\so^2,\so^3\sz\so^3, \so^7$& $41$\\
		$8$ & $\sz^8$, $\so\sz^6\so,\so\sz\so\sz^2\so\sz\so,\so^2\sz^4\so^2,\so^2\sz\so^2\sz\so^2,\so^3\sz^2\so^3, \so^8$ & $70$\\

		\cmidrule(lr){1-1} \cmidrule(l){2-2} \cmidrule(l){3-3} 
	\end{tabular}
	\makeatletter\def\@captype{table}\makeatother
	\caption{Prefix normal palindromes (\pnPal s).}
	\label{tab:prefixnormalpalindromes}
	\end{center}
\vspace*{-0.75cm}
\end{table}

\section{Properties of the Least-Representatives}

Before we present specific properties of the least representatives (\lr) for a 
given word length, we mention some useful properties of the maximum-ones, 
prefix-ones, and suffix-ones functions (for the basic properties we refer to 
\cite{journals/corr/abs-1805-12405, journals/tcs/BurcsiFLRS17} and the 
references therein). Since we are 
investigating only words of a specific length, we fix $n\in\N_0$.

Beyond the 
relation $p_w=s_{w^R}$ the mappings $p_w$ and $s_w$ are determinable from each 
other.
Counting the $\so$s in a suffix of length $i$ and adding the $\so$s in the 
corresponding prefix of length $(n-i)$ of a word $w$, gives the overall amount 
of 
$\so$s of $w$, namely
\[
p_w(n)=p_w(n-i)+s_w(i)\quad\mbox{and}\quad s_w(n)=p_w(i)+s_w(n-i).
\]
For suffix (resp. prefix) normal words this leads to 
$p_w(i)=f_w(n)-f_w(n-i)$ resp. $s_w(i)=f_w(n)-f_w(n-i)$ witnessing the fact 
$p_w=s_w$ for palindromes (since both equation hold). Before we show that 
indeed \pnPal s 
form a singleton class w.r.t. $\equiv_n$, we need the relation between the 
lexicographical order and prefix and suffix normality.

\begin{lemma}\label{leastsuffix}
The prefix normal form of a class is the lexicographically largest element 
in the class and the suffix-normal of a class is a \lr{}.
\end{lemma}

\begin{proof}
Let $w\in\Sigma^n$ be the prefix normal form of the class $[w]_{\equiv}$. 
Suppose there existed $v\in[w]_{\equiv}$ with $v>w$. Let $i\in[n]$ be the 
smallest index with $v[i]\neq w[i]$. Since we are only considering binary 
alphabets we get $v[i]=1$ and $w[i]=0$. By the prefix normality of $w$ we have 
$f_w(i)=p_w(i)=p_w(i-1)$ but on the other hand, $v\in[w]_{\equiv}$ and the 
minimality of $i$ implies 
\[
f_v(i)=f_w(i)=p_w(i-1)=p_v(i-1)=p_v(i)-1\leq f_v(i)-1<f_v(i).
\]
This contradiction shows that the prefix normal form of a class is the 
lexicographically largest element. The reverse $w^R$ is thus the 
lexicographically smallest element of the class which is by definition the 
$\lr${}. By Remark~\ref{sufpref} follows
\[
s_{w^R}=p_w=f_w=f_{w^R}
\]and hence the suffix 
normality of the $\lr${}.
\qed
\end{proof}

Lemma~\ref{leastsuffix} implies that a word being prefix and suffix normal 
forms a singleton class w.r.t. $\equiv_n$. As mentioned $p_w=s_w$ only holds 
for palindromes.

\begin{proposition}\label{corLol}
  For a word $w\in\Sigma^n$ it holds that $|[w]|_{\equiv}=1$ iff  $w\in\NPal(n)$. 
\end{proposition}

\begin{proof}
	Already in \cite{journals/tcs/BurcsiFLRS17}, the authors proved that
	$|[w]|_{\equiv}=1$ implies $w\in\NPal(n)$ for $w\in\Sigma^n$. The
	other direction follow immediately from 
	Lemma~\ref{leastsuffix}:  if
	$w$ is a prefix normal palindrome it is by definition prefix normal and by 
$w=w^R$, $w$ is the lexicographically largest and smallest element of the 
class. This implies that the class is a singleton.\qed
\end{proof}

The general part of this section is concluded by a 
somewhat artificial equation which is nevertheless useful for \pnPal s : by $s_w(i)=p_w^R(i)-p_w^R(i+1)+s_w(i-1)$ with $p_w^R(n+1)=0$ for 
$i\in[n]$ and $s_w=p_{w^R}$ we get
\[
p_{w^R}(i)=p_w^R(i)-p_w^R(i+1)-p_{w^R}(i-1).
\]
The rest of the section will cover properties of the \lr s of a 
class.

\begin{remark}\label{rem01}
For completeness, we mention that $\sz^n$ is the only even \lr~w.r.t. 
$\pnEquiv{n}$ and the only \pnPal{} starting with $\sz$.
Moreover, 
 $\so^n$ is the largest \lr{}. 
As we show later in the paper $\sz^n$ and $\so^n$ are of minor interest in the 
recursive process due to their speciality.
\end{remark}
The following lemma is an extension of \cite[Lemma 
1]{journals/tcs/BurcsiFLRS17} for the suffix-one function by relating the prefix 
and the suffix of the word $s_w$ for a 
least representative. Intuitively the suffix normality 
implies that the $\so$s are more at the end of the word $w$ rather than at 
the beginning: consider for instance $s_w=1123345$ for 
$w\in\Sigma^7$. The associated word $w$ cannot be suffix normal since the suffix of 
length two  has only one $\so$ ($s_w(2)=1$) but by $s_w(5)=3$, $s_w(6)=4$, and $s_w(7)=5$ we 
get that within two letters two $\so$s are present and consequently $f_w(2)\geq 
2$. Thus, a word $w$ is only least representative if the amount of $\so$s at 
the end of $s_w$ does not exceed the amount of $\so$s at the beginning of $s_w$.

\begin{lemma}\label{pchar}
Let $w\in\Sigma^n$ be a \lr{}. Then we have 
\[
s_{w}(i)\geq\begin{cases}
s_{w}(n)-s_{w}(n-i+1) & \mbox{if }s_{w}(n-i+1)=s_{w}(n-i),\\
s_{w}(n)-s_{w}(n-i+1)+1 & \mbox{otherwise}.
\end{cases}
\]
\end{lemma}

\begin{proof}
Since $w\in\Sigma^n$ is least representative we have 
$f_w(i)=s_w(i)=|\Suff_i(w)|_{\so}$ and 
$|\Suff_i(w)|_{\so}\geq|\Pref_i(w)|_{\so}$ for all 
$i\in[n]$. Let $i\in[n]$, $|\Suff_{n-i}(w)|_{\so}=s$, and 
$|\Pref_i(w)|_{\so}=r$.
This implies $s_w(n-i)=s$ and $s_w(n)=s+r$. If $s_w(n-i+1)=s_w(n-i)$ then 
$w[i]=\sz$ and $s_w(n-i+1)=s$. This implies 
\[
s_w(i)=|\Suff_i(w)|_{\so}\geq|\Pref_i(w)|_{\so}=r=s+r-s=s_w(n)-s_w(n-i+1).
\]
If $s_w(n-i+1)\neq s_w(n-i)$ then $w[i]=\so$, $s_w(n-i+1)=s+1$ and
\begin{align*}
s_w(i) & =|\Suff_i(w)|_{\so}\geq|\Pref_i(w)|_{\so}
=r=s+r-s=s_w(n)-s_w(n-i)\\
&=s_w(n)-s_w(n-i+1)+1.\tag*{\qed}
\end{align*}
\end{proof}

The remaining part of this section presents results for prefix normal 
palindromes. Notice that for $w\in\NPal(n)$ with $w=\sx v\sx$ with 
$\sx\in\Sigma$, $v$ is not necessarily a \pnPal{}; consider for 
instance $w=\so\sz\so\sz\so$ with $\sz\so\sz\in\Pal(3)\backslash\NPal(3)$. 
The following lemma shows a result for prefix normal 
palindromes which is folklore for palindromes substituting $f_w$ by $p_w$ or 
$s_w$.

\begin{lemma}\label{lempalrec}
For $w\in\NPal(n)\backslash\{\sz^n\}$, $v\in\Pal(n)$  with $w=\so 
v\so$ we have 
\[
f_w(k)= 
    \begin{cases}
        \so & \text{if $k = 1$,}\\
		f_v(k-1) +\so & \text{if $1 < k \leq |w|-1$,}\\
		f_w(|v|+1) + \so & \text{if $k = |w|$.}
	\end{cases}
\]
\end{lemma}

\begin{proof}
For $k=1$ we have $f_w(1)=|\Pref_1(w)|_1=1$ since $w\neq\sz^n$. For  
$k\in[|w|-1]_{>1}$ we have
\[
f_v(k-1)+\so=|\Pref_{k-1}(v)|_1+\so=|\Pref_k(w)|_1=p_w(k).
\]
Finally we have 
\begin{align*}
f_w(|v|+1)+\so&=|\Pref_{|v|+1}(w)|_1+\so=|\Pref_{|w|-1}
(w)|_1+\so=f_w(|w|-1)+\so\\
&=f_w(|w|).\tag*{\qed}
\end{align*}
\end{proof}

In the following we give a characterisation of when a palindrome $w$ is prefix 
normal  
depending on its maximum-ones function $f_w$ and a derived function $\overline{f_w}$. In particular we
observe that $f_w = \overline{f_w}^R$ if and only if $w$ is a prefix normal
palindrome. Intuitively $\overline{f_w}$ captures the progress of
$f_w$ in reverse order.  This is an intriguing result because it shows
that properties regarding prefix and suffix normality can be observed
when  $f_w,s_w,p_w$ are considered in their serialised representation.

\begin{definition} \label{fquer}
For $w\in\Sigma^{n}$ define $\overline{f}_w:[n]\rightarrow[n]$ by  
$\overline{f}_w(k)=\overline{f}_w(k-1)-(f_w(k-1)-f_w(k-2))$
with the extension $f_w(-1)=f_w(0)=0$ of $f$  and $\overline{f}_w(0) = f_w(n)$.
Define $\overline{p}_w$ and $\overline{s}_w$ analogously.
\end{definition}

\begin{example}
  Consider the \pnPal{} $w=\so\so\sz\so\so$ with
  $f_w=12234$.  Then $\overline{f}_w$ is $43221$ and we have
  $f_w=\overline{f}_w^R$. On the other hand for
  $v=\so\sz\so\so\sz\so\in\Pal(6)\backslash\NPal(6)$ we have
  $p_v=112334$ and $f_v = 122334$ and  $\overline{f}_v=432211$ and thus
  $\overline{f}_v^R\neq f_v$.  
\end{example}
The following lemma shows a connection
between the reversed prefix-ones function and the suffix-ones
function that holds for all palindromes.

\begin{lemma}\label{swopw}
	For $w \in \mathrm{Pal}(n)$ we have $s_w \equiv \overline{p}_w^R$.
\end{lemma}

\begin{proof}
	Let $w \in \mathrm{Pal}(n)$.  We 
get $\overline{p}_w^R(n) = p_w^R(1) = |w|_1 = s_w(n)$. Now let $i \in [n]_0$. 
Assume $s_w(n-i+1) = \overline{p}_w^R(n-i+1)$ holds. We have by induction
	\begin{align*}
	\overline{p}_w^R(n-i) &= \overline{p}_w(n-(n-i)+1) = \overline{p}_w(i+1)\\
	&= \overline{p}_w(i) - \left(p_w(i)-p_w(i-1)\right)\\
	&= \overline{p}_w^R(n-i+1) + (-s_w(i) + s_w(i-1))\\
	&= s_w(n-i+1)-w[i]\\
	&=s_w(n-i)+w[n-i+1]-w[i]\\
	&=s_w(n-i).\tag*{\qed}
	\end{align*}
\end{proof}

By Lemma~\ref{swopw} we get $p_w\equiv\overline{p}_w^R$ since $p_w\equiv s_w$ 
for a palindrome $w$.  As advocated earlier, our main theorem of this part
(Theorem~\ref{palchar}) gives a characterisation of \pnPal s. The theorem allows us to decide if a word is a \pnPal{} by only looking at the
maximum-ones-function, thus a comparison of all factors is not required.
\begin{theorem}\label{palchar}
Let $w \in \Sigma^n \setminus \Set{0^n}$. Then $w$ is a \pnPal{} if and only if $f_w = 
\overline{f}^\mathsf{R}_w$.
\end{theorem}

\begin{proof}
	Let $w\in\Sigma^n$. By definition of $\NPal(n)$, $w$ is prefix normal and a 
	palindrome, i.e. $s_w=f_w$. By Lemma~\ref{swopw} and Definition~\ref{fquer} we get 
	$f_w=s_w=\overline{p}_w^R=\overline{f}_w^R$. This proves $\Rightarrow$.
	Let $w \in \Sigma^n \setminus \Set{0^n}$ such that $f_w = 
	\overline{f}^\mathsf{R}_w$. If $w = \epsilon$, 
	then obviously $w \in \NPal(n)$ holds. Otherwise, if $w \neq \epsilon$, there exists a least representative $v \in \Sigma^n$ with $w \in [v]_{\equiv}$. This implies $f_v = f_w = \overline{f}_w^R = \overline{f}_v^R$, therefore the assumption also holds for $v$. Firstly, we will prove that $v$ is a palindrome. 
	Let $x \in \Set{\sz,\so}$ and $i \in [n]$. Thus we have $f_v(i-1) + x = f_v(i)$. Since $v$ is least representative this implies $v[i] = x$. By the assumption we get $\overline{f}_v^R(i-1) +x = \overline{f}_v^R(i)$ and applying the definition of the reversal and $\overline{f}_v$ leads to
	\begin{align*}
	\overline{f}_v(n - i + 1) &= \overline{f}_v^R(i)\\ 
	&= \overline{f}_v^R(i-1) +x\\
	&= \overline{f}_v(n - i +2 ) +x\\
	&= \overline{f}_v(n-i+1) - \left(f_v(n-i+1) - f_v(n-i)\right) +x.
	\end{align*}
	Hence we get $f_v(n-i+1) = f_v(n-i) + x$, i.e. $v[n-i+1] = x$. Thus $v[n-i+1] = v[i]$ and therefore $v$ is a palindrome. As proven in \cite{journals/tcs/BurcsiFLRS17}, prefix normal (and thus suffix 
	normal) palindromes are not pn-equivalent to any different word. 
Consequently $v = w$ and $w \in \NPal(n)$.  \qed
\end{proof}

\begin{table}[t]
  \centering
  \begin{tabularx}{\textwidth}{Y Y Y Y Y Y Y Y Y Y Y Y Y Y Y Y}	
	\scriptsize $i$ &\scriptsize  $1$ &\scriptsize  $2$ &\scriptsize  $3$ & \scriptsize $4$ & \scriptsize $5$ & \scriptsize $6$ & \scriptsize $7$& \scriptsize $8$ & \scriptsize $9$& \scriptsize $10$& \scriptsize $11$& \scriptsize $12$& \scriptsize $13$& \scriptsize $14$& \scriptsize $15$\\
	\cmidrule(r){2-16} 
	   \small\# & \small$2$ & \small$2$ &\small$3$ &\small$3$ &\small$5$ &\small$4$ &\small$8$ &\small$7$ &\small$12$ &\small$11$ &\small$21$ &\small$18$ &\small$36$ &\small$31$ &\small$57$\\
	 \cmidrule(r){2-16} 
\end{tabularx}
~\par~\par

\begin{tabularx}{\textwidth}{Y Y Y Y Y Y Y Y Y Y Y Y Y Y Y Y}	
	\scriptsize $i$ &\scriptsize  $16$ &\scriptsize  $17$ &\scriptsize  $18$ & \scriptsize $19$ & \scriptsize $20$ & \scriptsize $21$ & \scriptsize $22$& \scriptsize $23$ & \scriptsize $24$& \scriptsize $25$ & \scriptsize $26$ & \scriptsize $27$ & \scriptsize $28$ & \scriptsize $29$ & \scriptsize $30$\\
	\cmidrule(r){2-16} 
\small	\# &\small$55$ &\small$104$ &\small$91$ &\small$182$ &\small$166$ &\small$308$ &\small$292$ &\small$562$ &\small$512$ &\small$1009$ &\small $928$ & \small $1755$ & \small$1697$ & \small $3247$ & \small$2972$\\
	\cmidrule(r){2-16} 
\end{tabularx}
\caption{Number of \pnPal s. \cite{sloane2003line} (A308465)}
\label{tab:nbpalin}
\vspace*{-0.5cm}
\end{table}
\autoref{tab:nbpalin} presents the amount of \pnPal s up to 
length $30$.
These results support the conjecture in \cite{journals/tcs/BurcsiFLRS17} that there is a different behaviour for even and odd length of the word.

\section{Recursive Construction of Prefix Normal Classes}

In this section we investigate how to generate \lr s
of length $n+1$ using the the \lr s of length
$n$. This is similar to the work of
\citet{journals/corr/abs-1805-12405} except they investigated appending a 
letter to  prefix normal words while  we explore the behaviour on
prepending letters to \lr s. Consider the words 
$v = \so\sz\sz\so$ and $ w = \sz\sz\so\so$, both being 
(different) \lr s of length $4$. Prepending a $\so$
to them leads to $\so\so\sz\sz\so$ and $\so\sz\sz\so\so$ which are
 pn-equivalent. We say that $v$ and $w$ \emph{collapse} and denote
 it by
 $v\leftrightarrow w$.  
Hence for 
determining the index of $\equiv_n$ based on the least representatives of 
length $n-1$, only the least representative of one class matters. 

\begin{definition}
Two words $w,v\in\Sigma^n$ {\em collapse} if $\so w\pnEquiv{n+1}\so v$ holds. This 
is denoted by 
$w\leftrightarrow v$.
\end{definition}

\noindent  Prepending a $\so$ to a non \lr{} will never lead to a \lr{}. Therefore It is sufficient to only look at \lr s. Since collapsing is an equivalence relation, denote the equivalence 
class 
w.r.t. $\leftrightarrow$ of a word $w\in\Sigma^{\ast}$ by 
$[w]_{\leftrightarrow}$. Next, we present some general results regarding the 
connections between 
the \lr s of lengths $n$ and $n+1$.
As mentioned in Remark~\ref{rem01}, $\sz^n$ and $\so^n$ are for all $n\in\N$ 
\lr s. This implies that they do not have to be considered in 
the recursive process. 

\begin{remark}\label{prepend0}
By \cite{journals/corr/abs-1805-12405} a word  
$w\sz\in\Sigma^{n+1}$ is prefix-normal 
if $w$ is prefix-normal. Consequently we know that if a word $w\in\Sigma^n$ is 
suffix normal, $\sz w$ is suffix normal as well. This leads in accordance to the 
naïve upper bound of $2^{n}+1$ to a naïve lower bound of 
$|\Sigma^n/\equiv_n|$ for $|\Sigma^{n+1}/\equiv_{n+1}|$. 
\end{remark}

\begin{remark}\label{trivialextensionresults}
The maximum-ones functions for $w\in\Sigma^{\ast}$ and $\sz w$ are equal on 
all $i\in[|w|]$ and $f_{\sz w}(|w|+1)=f_w(|w|)$ since the factor determining 
the maximal number of $\so$'s is independent of the leading $\sz$. Prepending  
$\so$ to a word $w$ may result in a difference between $f_w$ and $f_{\so 
w}$, but notice that since only one $\so$ is prepended, we always have $f_{\so 
w}(i)\in\{f_{w}(i),f_{w}(i)+1\}$ for all $i\in [n]$. In both cases we have 
$s_w(i)=s_{\sx w}(i)$ for $\sx\in\{\sz,\so\}$ and $i\in[|w|]$ and $s_{\sz 
w}(n+1)=s_w(n)$ as well as $s_{\so w}(n+1)=s_w(n)+1$.
\end{remark}
Firstly we improve the naïve upper bound to $2|\Sigma^{n}/\equiv_n|$ by proving that only \lr s in $\Sigma^n$ can become \lr s in $\Sigma^{n+1}$ by prepending $\so$ or $\sz$.

\begin{proposition}\label{non-least}
Let $w\in\Sigma^n$ not be \lr{}. Neither $\sz w$ nor $\so 
w$ are \lr s in $\Sigma^{n+1}$.
\end{proposition}

\begin{proof}
Suppose firstly $\sz w$ is a least representative, i.e. $f_{\sz w}(i)=s_{\sz 
w}(i)$ for all $i\in[n+1]$. By $s_{\sz w}(i)=s_{w}(i)$ and 
$f_{\sz w}(i)=f_{w}(i)$ for $i\in[n]$, we have $s_w(i)=f_w(i)$ and thus $w$ 
would be a least representative. Suppose that secondly $\so w$ is a least 
representative. Since $w$ is not a least representative there exists a 
$j\in[|w|]$ with $s_w(j)\neq f_w(j)$. Choose $j$ minimal.
Since $\so w$ is a least representative, we get 
\[
f_{\so w}(j)=s_{\so w}(j)=s_w(j)\neq f_w(j).
\]
By Remark~\ref{trivialextensionresults} we have $f_{\so w}(j)=f_w(j)+1$
$s_w(j)\leq f_w(j)$ implies 
$f_w(j)\geq s_w(j)=f_{\so w}(j)=f_w(j)+1$ - a contradiction. \qed
\end{proof}

By Proposition~\ref{corLol} prefix (and thus suffix)
normal palindromes form a singleton class. This 
implies immediately that a word $w\in\Sigma^n$ such that $\so w$ is a prefix 
normal palindrome, does not collapse with any other 
$v\in\Sigma^n\backslash\{w\}$. The next lemma shows that even prepending once a 
$\so$ and once a $\sz$ to different words leads only to equivalent words in one 
case.

\begin{lemma}\label{falsecollapse}
Let $w,v\in\Sigma^n$ be different \lr s. Then $\sz w\equiv_n\so v$ if and only 
if $v=\sz^n$ and $w=\sz^{n-1}\so$. 
\end{lemma}

\begin{proof}
The equivalence of $\sz \sz^{n-1}\so=\sz^n\so$ and $\so\sz^{n}$ is immediate. This proves the $\Leftarrow$-direction. For the other direction assume $\sz w\equiv_n\so v$. By definition we get $f_{\sz w}(i)=f_{\so v}(i)$ for all $i\in[n+1]$ and moreover by Remark~\ref{trivialextensionresults} $f_w(i)=f_{\so v}(i)$ for all $i\in[n]$. 
By $s_w(1)=f_w(1)=f_{\so v}(1)=1$ we get $w[|w|]=\so$ and by $s_w(n)=f_{\so v}(n)=|w|_{\so}$ there exists $u\in\Fact_n(\so v)$ with $|u|_{\so}=|w|_{\so}$. The equivalence of $\sz w$ and $\so v$ implies $|w|_{\so}=|v|_{\so}+1$ and thus $u$ has to be a prefix of $\so v$. Hence $u[2 .. n]$ is a prefix of $v$ of length $n-1$ with $|w|_{\so}-1$ $\so$s. Since this is the overall amount of $\so$ in $v$, $v[n]=0$ follows and implies immediately $v=\sz^n$. By $s_w(1)=f_{\so v}(1)=1$ follows $w[n]=\so$ and the claim follows with $|\sz w|_{\so}=|\so v|_{\so}$.\qed
\end{proof}

By Lemma~\ref{falsecollapse} and Remark~\ref{prepend0} it suffices to 
investigate the collapsing relation on prepanding $\so$s.  The following 
proposition characterises the \lr{} $\so w$ among the elements $\so v\in[\so 
w]_\equiv$ for all \lr s $v\in\Sigma^{n}$ with $w\leftrightarrow v$ for 
$w\in\Sigma^n$.

\begin{proposition}\label{minrepext}
Let $w\in\Sigma^n$ be a \lr{}.
Then $\so w\in\Sigma^{n+1}$ is a \lr{} if and only if $f_{1w}(i)=f_w(i)$ holds for $i\in[n]$ and $f_{1w}(n+1)=f_w(n)+1$. 
\end{proposition}

\begin{proof}
Let $w\in\Sigma^n$ be a least representative.
Consider firstly $\so w\in\Sigma^{n+1}$ to be least representatives as well. Since $w$ is a least representative we have $s_w(i)=f_w(i)$ for all $i\in[n]$. By Remark~\ref{trivialextensionresults} follows $s_w(i)=s_{\so w}(i)$ for all $i\in[n]$ and with $\so w$ being a least representative we get $f_w(i)=f_{\so w}(i)$ for all $i\in[n]$. By the same arguments we get $f_{\so w}(n+1)=s_{\so w}(n+1)=s_{w}(n)+1=f_w(n)+1$. Similarly we get for the second direction $f_{\so w}(i)=f_w(i)=s_w(i)=s_{\so w}(i)$ for all $i\in[n]$ and $f_{\so w}(n+1)=f_w(n)+1=s_w(n)+1=s_{\so w}(n+1)$.\qed
\end{proof}

\begin{corollary}\label{minrepextpal}
Let $w\in\NPal(n)$. Then $f_{w\so}(i)=f_w(i)$ for $i\in[n]$ and 
$f_{w\so}(n+1)=f_w(n)+1$. Moreover $s_{w\so}(i)=s_w(i)$ for $i\in[n]$ and 
$s_{w\so}(n+1)=s_w(n)+1$.
\end{corollary}

\begin{proof}
Since $w$ is a prefix normal palindrome, we have $(\so w)^R=w^R\so=w\so$. This implies 
$f_w(i)=f_{\so w}(i)=f_{(\so w)^R}(i)=f_{w\so}(i)$ for all $i\in[n]$ and 
$f_{w\so}(n+1)=f_w(n)+1$. If 
$s_{w\so}(i)=s_{w}(i)+1$ then 
$s_{w\so}(i)=s_{w}(i)+1=f_w(i)+1=f_{w\so}(i)+1$ would contradict 
$s_{w\so}(i)\leq 
f_{w\so}(i)$ for some $i\in[n+1]$. This proves the claim for the suffix-ones 
function.\qed
\end{proof}

This characterization is unfortunately not convenient for determining either 
the number of \lr s of length $n+1$ from the ones from length 
$n$ or the collapsing \lr s of length $n$. For a given word 
$w$, the maximum-ones function $f_w$ has to be determined, $f_w$ to be extended 
by $f_w(n)+1$, and finally the associated word - under the assumption $f_{\so 
w}\equiv s_{\so w}$ has to be checked for being suffix normal. For instance, 
given $w=\so\sz\sz\so\sz\so$ leads to $f_w=11223$, and is extended to $f_{\so 
w}=112234$. This would correspond to $\so\so\sz\so\sz\so$ which is not suffix 
normal and thus $w$ is not extendable to a new \lr{}. The 
following two lemmata reduce the amount of \lr s that needs to 
be checked for extensibility.

\begin{lemma}\label{smallsum}
Let $w\in\Sigma^n$ be a \lr{} such that $\so w$ is a \lr{} as well. Then for all \lr s 
$v\in\Sigma^n\backslash\{w\}$ collapsing with $w$, $f_v(i)\leq f_w(i)$ holds 
for all $i\in[n]$, i.e. all other \lr s have a 
smaller maximal-one sum.
\end{lemma}

\begin{proof}
Let $v\in\Sigma^n\backslash\{w\}$ a least representative with $w\leftrightarrow_n v$. 
By the property of being least representatives, the definition of the maximum-ones and suffix-ones functions follows for all $i\in[n]$
\[
f_v(i)=s_v(i)=s_{\so v}(i)\leq f_{\so v}(i)=f_{\so w}(i)=s_{\so w}(i)=s_w(i)=f_w(i).
\]
By $v\neq w$ there exists at least one $j\in[n]$ with $s_w(j)>s_v(j)$. This implies 
\begin{align*}
\sigma(v)&=\sum_{i\in[n]}f_v(i)=\sum_{i\in[n]}s_v(i)\\
&<s_w(j)+\sum_{i\in[n]\backslash\{j\}} s_w(j)=\sum_{i\in[n]}s_w(j)=\sum_{i\in[n]}f_w(j)\\
&=\sigma(w).\qed
\end{align*}
\end{proof}

\begin{corollary}\label{lexsmall}
If $w,v\in\Sigma^n$ and $\so w\in\Sigma^{n+1}$ are \lr s with 
$w\leftrightarrow v$ and $v\neq w$ then $w\leq v$. 
\end{corollary}

\begin{proof}
By $v\neq w$ exists an $i\in[n]$ minimal with $w[i]\neq v[i]$. Suppose $w[i]=1$ and $v[i]=0$.
By $f_{\so v}\equiv f_{\so w}$ we get $|w[i+1 .. n]|_{\so}+1=|v[i+1 .. n]|_{\so}$. Thus
\[
s_{\so v}(n-i)=s_{v}(n-i)=s_w(n-i)+1=f_w(n-i)+1=f_{\so w}(n-i)+1=f_{\so v}(n-i)+1.
\]
This contradicts $s_{\so v}(n-i)\leq f_{\so v}(n-i)$.\qed
\end{proof}

\begin{remark}\label{extendible}
By Corollary~\ref{lexsmall} the lexicographically smallest \lr{} 
$w$ among the collapsing leads to the \lr{} of $[\so w]$. Thus 
if $w$ is a \lr{} not collapsing with any lexicographically 
smaller 
word then $\so w$ is \lr{}.
\end{remark}
Before we present the theorem characterizing exactly the collapsing words for a 
given word $w$, we show a symmetry-property
of the \lr s which are not extendable to \lr s, i.e. a property of words which collapse.

\begin{lemma}\label{symminf}
Let  $w\in\Sigma^n$ be a \lr{}. Then $f_{1w}(i)\neq f_w(i)$ for 
some $i\in[n]$ iff
$f_{1w}(n-i+1)\neq f_w(n-i+1)$.
\end{lemma}

\begin{proof}
Since $w$ is least representative, we have 
\[
f_w(n-i+1)=s_w(n-i+1)=|\Suff_{n-i+1}(w)|_{\so}=|w|_{\so}-|\Pref_{i-1}(w)|_{\so}. 
\]
From $f_w(i)\neq f_{\so w}(i)$ follows $f_{\so w}(i)=f_w(i)+1=s_w(i)+1$. Thus $\so w$ has a factor of length $i$ with $s_w(i)+1$ $\so$s. The suffix normality of $w$ implies that this factor needs to be the prefix of $\so w$ of length $i$, i.e. $|\Pref_{i-1}(w)|_{\so}=s_w(i)$. Thus we get $f_w(n-i+1)=|w|_{\so}-s_w(i)$. On the other hand we have
\[
|\Pref_{n-i+1}(\so w)|_{\so}=|\Pref_{n-i}(w)|_{\so}+1=|w|_{\so}-|\Suff_i(w)|_{\so}+1=|w|_{\so}-s_w(i)+1.
\]
Consequently $f_{\so w}(n-i+1)\geq|w|_1-s_w(i)+1>f_w(n-i+1)$. The second 
direction follows immediately with $j:=n-i+1$ and $f_{\so w}(j)\neq f_w(j)$. 
\qed
\end{proof}

By \cite[Lemma 10]{journals/tcs/BurcsiFLRS17} a word $w\so$ is prefix normal if 
and only if $|\Suff_k(w)|_{\so}<|\Pref_{k+1}(w)|_{\so}$ for all $k \in \mathbb{N}$. The following 
theorem extends this result for determining the collapsing words $w'$ for a 
given word $w$.

\begin{theorem}\label{collapstheo}
Let $w\in\Sigma^n$ be a \lr{} and $w'\in\Sigma^n\backslash\{w\}$ 
with 
$|w|_1=|w'|_1=s\in\N$. Let moreover $v\not\leftrightarrow w$ for all 
$v\in\Sigma^{\ast}$ with $v\leq w$. Then $w\leftrightarrow w'$ iff
\begin{enumerate}
\item\label{p1} $f_{w'}(i)\in\{f_w(i),f_w(i)-1\}$ for all $i\in[n]$,
\item\label{p2} $f_{w'}(i)=f_w(i)$ implies $f_{1w'}(i)=f_{w}(i)$,
\item\label{p3} $f_{w'}(i)\geq\begin{cases}
f_{w'}(n)-f_{w'}(n-i+1) & \mbox{if }f_{w'}(n-i+1)=f_{w'}(n-i),\\
f_{w'}(n)-f_{w'}(n-i+1)+1 & \mbox{otherwise}.
\end{cases}$
\end{enumerate}
\end{theorem}

\begin{proof}
Notice that $|w|_1=|w'|_1=s\in\N$ implies immediately $f_w(1)=f_{w'}(1)=1$ and 
$f_{w'}(n)=f_w(n)=s$. Moreover for all $i\in[n]$ by Lemma~\ref{f_w(k+1)} we 
have 
$f_{w'}(i)\in\{f_{w'}(i-1),f_{w'}(i-1)+1\}$ and 
$f_{w}(i)\in\{f_{w}(i-1),f_{w}(i-1)+1\}$ and by Lemma~\ref{symminf} we get 
$f_{w'}(i)\neq f_{1w'}(i)$ iff $f_{w'}(n-i+1)\neq f_{1w'}(n-i+1)$.

\medskip

Firstly consider the $\Leftarrow$-direction, i.e. let $w'\in\Sigma^n$ with 
$|w|_1=s$ and the properties \ref{p1}, \ref{p2}, and \ref{p3}. We have to prove 
$w'\leftrightarrow w$, 
hence we have to prove $f_{\sz w}(i)=f_{\sz w'}(i)$ for all $i\in[n]$. Since 
$w$ does not collapse with any lexicographically smaller $v\in\Sigma^n$, $\sz 
w$ is a least representative by Remark~\ref{extendible}. From
Proposition~\ref{minrepext} follows
$f_w(i)=f_{1w}(i)$ for all $i\in[n]$. Obviously we have $f_{\so
w}(1)=1=f_{\so w'}(1)$ and hence the claim holds for $i=1$. By 
$f_{w'}(n)=s$ we get $f_{1w'}(n)\in\{s,s+1\}$. If $f_{1w'}(n)$ were $s+1$ then 
by $f_{1w'}(n)\neq f_{w'}(n)$ and consequently by Lemma~\ref{symminf} we would 
have
$1=f_{1w'}(1)\neq f_{w'}(1)=1$. Hence the claim holds for $i=n$. Let 
$i\in[n-1]_{>1}$. The claim holds by property 1 and 
Proposition~\ref{minrepext} if $f_{w'}(i)=f_w(i)$. Hence, 
assume $f_{w'}(i)\neq f_w(i)$ for an $i\in[n-1]_{>1}$, i.e. 
$f_{w'}(i)=f_w(i)-1$ by property 1.  By Remark~\ref{trivialextensionresults} we 
have $f_{\so w'}(i)\in \{f_{w'}(i),f_{w'}(i)+1\}$.\\
{\bf case 1:} $f_{w'}(i)=f_{1w'}(i)$ \\
If $w'$'s prefix of length $i-1$ had more 
(or equal) $\so$s than the suffix of length $i$, then the prefix of $\so w'$ of 
length $i$ would have strictly more $\so$s than the suffix of length $i$. This 
contradicts $f_{w'}(i)={1w'}(i)$ and thus we have 
$|\Pref_{i-1}(w')|_{\so}<|\Suff_i(w')|_{\so}$. By 
$f_{w'}(n-i+1)\geq|\Suff_{n-i+1}(w')|_{\so}$ and 
$|\Suff_{n-i+1}(w')|_{\so}+\|\Pref_{i-1}(w')|_{\so}=s$ we get
\begin{align*}
s-f_{w'}(n-i+1) & \leq 
s-|\Suff_{n-i+1}(w')|_{\so}=|\Pref_{i-1}(w')|_1\\
&<|\Suff_i(w')|_{\so}\leq 
f_{w'}(i).
\end{align*}
This is a contradiction to property 3.\\
{\bf case 2:} $f_{w'}(i)+1=f_{1w'}(i)$\\
In this case we get immediately
\[
f_{\so w'}(i)=f_{w'}(i)+1=f_w(i)-1+1=f_w(i)=f_{\so w}(i).
\]
Thus 
$f_{1w'}(i)=f_w(i)$ for all $i\in[n]$ which means that $f_{1w'}$ and $f_{1w}$ 
are identical, i.e. $w$ and $w'$ collapse.

\medskip

For the $\Rightarrow$-direction, assume $w\leftrightarrow w'$, i.e.
$f_{1w'}= f_{1w}$. Proposition~\ref{non-least} implies that $w'$ can be 
assumed as a least representative since $\so w$ is a least representative and 
by $w$ and $w'$ collapsing, $\so w'$ is one as well. By 
Proposition~\ref{minrepext} we have 
$f_{1w}(i)=f_w(i)$ 
for all $i\in[n]$ and thus $f_{1w'}(i)=f_w(i)$ which proves (\ref{p2}). Since 
$f_w(i)=f_{1w'}(i)\in\{f_{w'}(i),f_{w'}(i)+1\}$ for all $i\in[n]$ we get 
property \ref{p1}. Since $w'$ is a least representative,  Lemma~\ref{pchar} 
implies property \ref{p3}.\qed
\end{proof}

Theorem~\ref{collapstheo} allows us to construct the equivalence classes w.r.t. 
the least representatives of the previous length but more tests than necessary 
have to be performed: Consider, for 
instance $w=\so\so\so\sz\so\so\sz\sz\so\so\so\sz\so\so\so\so\so$ which is a smallest \lr{} 
of length $17$ not collapsing with any lexicographically smaller \lr{}. For $w$ 
we have $f_w=1\cdot 2\cdot 3\cdot 4\cdot 5\cdot 
5\cdot 6\cdot 7\cdot 8\cdot 8\cdot 8\cdot 9\cdot 10\cdot 10\cdot 11\cdot 
12\cdot 13$ where the dots just act as separators between letters. 
Thus we know for any $w'$ 
collapsing with $w$, that $f_{w'}(1)=1$ and $f_{w'}(17)=13$. The constraints
$f_{w'}(2)\in\{f_{w'}(2),f_{w'}(2)+1\}$ and $f_{w'}(2)\leq f_{w}(2)$ implies 
$f_{w'}(2)\in\{1,2\}$. First the check that $f_{w'}(10)=4$ is impossible 
excludes $f_{w'}(2)=1$. Since no collapsing word can have a factor of length 
$2$ with only one $\so$, a band in which the possible values range can be 
defined by the unique greatest collapsing word $w'$. It is not surprising that 
this word is connected with the prefix normal form.
The following two lemmata define the band in which the possible collapsing 
words 
$f_w$ are.

\begin{lemma}\label{schlauch1}
Let $w\in\Sigma^n\backslash\{\sz^n\}$ be a \lr{} with 
$v\not\leftrightarrow w$ 
for all $v\in\Sigma^n$ with $v\leq w$. Set $u:=(\so w[1..n-1])^R$. Then 
$w\leftrightarrow u
$ and for all \lr s $v\in\Sigma^n\backslash\{u\}$ with 
$v\leftrightarrow w$ and all $i\in[n]$ $f_{v}(i)\geq f_{u}(i)$, thus 
$\sigma(u) = \sum_{i\in[n]}f_u(i) \leq \sum_{i\in[n]}f_v(i) = \sigma(v)$.
\end{lemma}

\begin{proof}
Set $u=(1w[1..n-1])^R$. Then by $w$ odd follows
\begin{align*}
f_{\so w} &= f_{(\so w)^R} = f_{w^R \so} = 
f_{w[n](w[1..n-1])^R\so } = f_{w[n](\so w[1..n-1])^R} = f_{\so 
(1w[1..n-1])^R}\\
&= f_{\so u},
\end{align*}
i.e. $w\leftrightarrow u$. Since $w$ does not collapse with any 
lexicographically smaller word, $\so w$ is a least representative by 
Remark~\ref{extendible}. By Remark~\ref{sufpref} $(\so w)^R\in[\so w]_{\equiv}$ and 
$w^R\so$ is lexicographically the largest element in the class. If there 
existed a $v\in[w]_{\leftrightarrow}$ with $v> w[n-1..1]\so$, then
\[
\so v> \so w[n-1..1]=w^R\so=(\so w)^R
\]
would hold which contradicts the maximality of $(\so w)^R$.\qed
\end{proof}

Notice that $w'=(\so w[1..n-1])^R$ is not necessarily a \lr{} in 
$\Sigma^n/\equiv_n$ witnessed by the word of the last example. For $w$ we get
$u=\so\so\so\sz\so\so\so\sz\sz\so\so\sz\so\so\so\so$ with $f_{u}(8)=f_w(8)$ and $f_{u}(10)=7\neq 8=f_w(10)$ 
violating the symmetry property given in Lemma~\ref{symminf}. The following 
lemma alters $w'$ into a \lr{} which represents still the lower 
limit of the band.

\begin{lemma}\label{schlauch2}
Let $w\in\Sigma^n$ be a \lr{} such that $\so w$ is also a \lr{}. Let $w'\in\Sigma^n$ with $w\leftrightarrow w'$, and $I$ the set 
of all $i\in[\lfloor\frac{n}{2}\rfloor]$ with 
\begin{align*}
(f_{w'}(i)=f_w(i) &\land f_{w'}(n-i+1)\neq f_w(n-i+1))\mbox{ or } \\
(f_{w'}(i)\neq f_w(i) &\land f_{w'}(n-i+1)= f_w(n-i+1))
\end{align*}
and $f_w(j)=f_{w'}(j)$ for all $j\in[n]\backslash I$. Then $\hat{w}$ defined 
such that $f_{\hat{w}}(j)=f_{w'}(j)$ for all $j\in[n]\backslash I$ and 
$f_{\hat{w}}(n-i+1)=f_{w'}(n-i+1)+1$ ($f_{\hat{w}}(i)=f_{\hat{w}}(i)+1$ 
resp.) 
for all $i\in I$ holds, collapses with $w$.
\end{lemma}

\begin{proof}
Let $k\in[n]$. Since $\so w$ is least representative, we have 
$f_{\so w}(k)\geq f_{\so \hat{w}}(k)$. If $k\not\in I$ we get
\[
f_{\so w}(k)=f_w(k)=f_{w'}(k)=f_{\hat{w}}(k)\leq f_{\so\hat{w}}(k)
\]
and thus $f_{\so w}(k)=f_{\so \hat{w}}(k)$. If $k\in I$ we get in the first case
\[
f_{\so w}(n-k+1)=f_w(n-k+1)=f_{w'}(n-k+1)+1=f_{\hat{w}}(n-k+1)\leq 
f_{\so\hat{w}}(n-k+1)
\]
and thus $f_{\so w}(n-k+1)=f_{\so\hat{w}}(n-k+1)$. This second case holds 
analogously.\qed
\end{proof}

\begin{remark}
Lemma~\ref{schlauch2} applied to $(\so w[1..n-1])^R$ gives the lower limit of the 
band. Let $\hat{w}$ denote the output of this application for a given 
$w\in\Sigma^n$ according to Lemma~\ref{schlauch2}.
\end{remark}
Continuing with the example, we firstly determine $\hat{w}$ for 
$w=\so\so\so\so\sz\so\so\so\sz\sz\so\so\sz\so\so\so\so$. We get with $u=w[n-1..1]1$
Since for all collapsing $w'\in\Sigma^n$ we have $f_{\hat{w}}(i)\leq 
f_{w'}(i)\leq f_w(i)$, $w'$ 
is determined for $i\in[17]\backslash\{5,9,13\}$.
Since the value for $5$ de-
	
\begin{center}
	\begin{tabular}{@{}lccccccccccccccccc@{}}
		
		\scriptsize\color{black!70} i        & \scriptsize\color{black!70} 1 
&\scriptsize\color{black!70} 2 &\scriptsize\color{black!70} 3 
&\scriptsize\color{black!70} 4 &\scriptsize \color{black!70}5 
&\scriptsize\color{black!70} 6 & \scriptsize\color{black!70} 7 
&\scriptsize\color{black!70} 8 &\scriptsize\color{black!70} 9 & 
\scriptsize\color{black!70}10 &\scriptsize\color{black!70} 11 & 
\scriptsize\color{black!70}12 & \scriptsize\color{black!70}13 & 
\scriptsize\color{black!70}14 & \scriptsize\color{black!70}15 & 
		\scriptsize\color{black!70}16 & \scriptsize\color{black!70}17 
\\\cmidrule[0.8pt](lr){2-18} 
		$f_w$    & 1 & 2 & 3 & 4 & 5 & 5 & 6 & 7 & 8 & 8 & 8 & 9 & 10 & 10 & 11 
& 12 & 
		13  \\ \cmidrule(lr){2-18} 
		$f_{u}$  & 1 & 2 & 3 & 4 & 4 & 5 & 6 & 7 & 7 & 7 & 8 & 9 & 9 & 10 & 11 & 
12 & 
		13 \\ \cmidrule(lr){2-18} 
		$f_{\hat{w}}$  & 1 & 2 & 3 & 4 & 4 & 5 & 6 & 7 & 7 & 8 & 8 & 9 & 9 & 10 
& 11 & 
		12 & 13 \\
		\cmidrule[0.8pt](lr){2-18} 
	\end{tabular}
	\makeatletter\def\@captype{table}\makeatother
	\caption{$f$ for $w=\so\so\so\so\sz\so\so\so\sz\sz\so\so\sz\so\so\so\so$.}
	\label{tab:examplew}
	\end{center}
\noindent
termines the one for $13$ there are only two 
possibilities, namely $f_{w'}(5)=5$ and $f_{w'}(9)=7$ and $f_{w'}(5)=4$ and 
$f_{w'}(9)=8$. Notice that the words $w'$ corresponding to the generated 
words $f_{w'}$ are not necessarily \lr s of the shorter 
length as witnessed by the one with $f_{w'}(5)=5$ and $f_{w'}(9)=7$. In this 
example this leads to at most three words being not only in the class 
but also in the list of former representatives. Thus we are able to produce an 
upper bound for the cardinality of the class.
Notice that in any case we only have to test the 
first half of $w'$'s positions by Lemma~\ref{symminf}. 
This leads to the following definition.

\begin{definition}
Let $h_d:\Sigma^{\ast}\times\Sigma^{\ast}\rightarrow\N_0$ be the 
Hamming-distance. The {\em palindromic distance} 
$p_d:\Sigma^{\ast}\rightarrow\N_0$ is defined by 
$p_d(w)=h_d(w[1..\lfloor\frac{n}{2}\rfloor],(w[\lceil\frac{n}{2}\rceil+1..|w|]
)^R )$. Define the {\em palindromic prefix length} 
$p_{\ell}:\Sigma^{\ast}\rightarrow\N_0$ by 
$p_{\ell}(w)=\max\left\{\,k\in[|w|]\,|\,\exists 
u\in\Pref_k(w):\,p_d(u)=0\,\right\}$.
\end{definition}
The palindromic distance gives the minimal number of positions in which a bit 
has to be 
flipped for obtaining a palindrome. Thus, $p_d(w)=0$ for all palindromes $w$, 
and, for instance, $p_d(\so\so\sz\sz\so\so\sz\sz\so)=2$ since the first half of 
$w$ and the reverse of the second half mismatch in two positions.
The palindromic prefix length determines the length of $w$'s longest prefix 
being a palindrome. For instance $p_{\ell}(\so\so\sz\so)=2$ and 
$p_{\ell}(\sz\so\so\sz\so)=4$.
Since  a \lr{} $w$ determines the upper limit of the band and 
$w[n-1..1]\so$ the lower limit, the palindromic distance of $ww[n-1..1]\so$ is 
in relation to the positions of $f_w$ in which collapsing words may differ from 
$w$.

\begin{theorem}\label{collapsindex}
If $w\in\Sigma^n$ and $\so w$ are both \lr s then  $|[w]_{\leftrightarrow}|
\leq 
2^{\lceil\frac{p_d(ww[n-1..1]\so}{2}\rceil}$.
\end{theorem}

\begin{proof}
By Lemma~\ref{schlauch1} $w[n-1..1]\so$ determines the lower bound of the band 
for collapsing words. Let $s_1,\dots,s_{\ell}\in[n]$ with $s_i<s_{i+1}$ for 
$i\in[\ell], \ell\in[n]$ be the positions with $w[s_i]\neq 
(w[n-1..1]\so)[s_i]$. Thus for all odd $i\in[\ell-1]$, $f_w$ and 
$f_{w[n-1..1]\so}$ are different between $s_i$ and $s_{i+1}-1$, since a 
different bit leads to a different number of $\so$s. By the same argument, 
$f_w$ and $f_{w[n-1..1]\so}$ are identical between $s_{i}$ and $s_{i+1}-1$ for 
all even $i\in[\ell-1]$. This implies that only the differences in odd 
positions lead to differences values of the corresponding maximum-ones 
function. Since each difference in the maximum-ones functions can be altered 
independently for obtaining a potential collapsing word, the number of 
collapsing words is exponential in half the palindromic distance. \qed
\end{proof}

For an algorithmic approach to determine the \lr s of length 
$n$, we want to point out that the search for collapsing words can also be reduced using the palindromic prefix length. Let $w_1,\dots, w_m$ be 
the \lr s of length $n-1$. For each $w$ we keep track of 
$|w|-p_{\ell}(w)$. 
For each $w_i$ we check firstly if $|w_i|-p_{\ell}(w_i)=1$ since in this case 
the prepended $\so$ leads to a palindrome. Only if this is not the case, 
$[w_i]_{\leftrightarrow}$ needs to be determined. All collapsing words computed 
within the band of $w_i$ and $\hat{w_i}$ are deleted in $\{w_{i+1},\dots,w_m\}$.

\medskip

In the remaining part of the section we investigate the set $\NPal(n)$ w.r.t. 
$\NPal(\ell)$ for $\ell<n$. This leads to a second calculation for an upper 
bound and a refinement for determining the \lr s of 
$\Sigma^n/\equiv_n$ faster.

\begin{lemma}\label{notpal}
If $w\in\NPal(n)\backslash\{\so^n\}$ then $\so w$ is not a \lr{} 
but $w\so$ is a \lr{}.
\end{lemma}

\begin{proof}
It suffices to prove that $w\so$ is a least representative. Then $(w\so)^R=\so 
w$ is prefix normal and since $w\so$ is not a palindrome, $\so w$ is not a 
least representative. By Corollary~\ref{minrepextpal} follows immediately that 
$w\so$ is a least representative.\qed
\end{proof}

\begin{remark}\label{palupperbound}
By Lemma~\ref{notpal} follows that all words $w\in\NPal(n)$ collapse with a 
smaller \lr{}. Thus, for all $n\in\N$, an upper bound for 
$|\Sigma^{n+1}/\equiv_{n+1}|$ is given by 
$2|\Sigma^n/\equiv_n|-\npal(n)$.
\end{remark}
For a closed recursive calculation of the upper 
bound in Remark~\ref{palupperbound}, the exact number $\npal(n)$ is 
needed. Unfortunately we are not able to determine $\npal(n)$ for arbitrary 
$n\in\N$. The following results show relations between prefix normal 
palindromes of different lengths. For instance, if $w\in\NPal(n)$ then $1w1$ is 
a prefix normal palindrome as well.
The importance of the the \pnPal s is witnessed by the 
following estimation.

\begin{theorem}\label{palcol}
For all $n\in\N_{\geq 2}$ and $\ell=|\Sigma^n/\equiv_n|$ we have
\[
\ell+\npal(n-1)\leq|\Sigma^{n+1}/\equiv_{n+1}|\leq 
\ell+\npal(n+1)+\frac{\ell-\npal(n+1)}{2}.
\]
\end{theorem}

\begin{proof}
By Corollary~\ref{minrepextpal}, $w\so$ is a least representative. From 
Lemma~\ref{schlauch1} follows that the band for possible collapsing words is 
given by $(\so (w\so)[1..n])^R=(\so w)^R=w^R1=w1$. If the band is empty, 
$\so w\so$ is the single element in $[\so w\so]$ and consequently a 
palindrome \cite{journals/corr/abs-1805-12405}. This implies the lower bound.  For the upper bound firstly all 
$\sz w$, for $w$ being a least representative in $\Sigma^n/\equiv_n$ have to be
counted and secondly all prefix normal palindromes. All other elements collapse 
with at least one different element.\qed
\end{proof}

The following results only consider \pnPal s that are different 
from $\sz^n$ and $\so^n$. Notice for these special palindromes that 
$\sz^n\sz^n$, $\so^n\so^n$, $\so^n\so\so^n$, $\sz^n\sz\sz^n$, 
$\so\so^n\so^n\so$, $\so\sz^n\sz^n\so\in\NPal(k)$ for an appropriate $k\in\N$ 
but 
$\sz^n\so\sz^n\not\in\NPal(2n+1)$.

\begin{lemma}\label{ww}
If $w\in\NPal(n)\backslash\{\so^n,\sz^n\}$ then neither $ww$ 
nor $w\so w$ are prefix normal palindromes.
\end{lemma}

\begin{proof}
Let $k\in[n-1]$ be minimal with $f_w(k)=f_w(k+1)$ (exists by 
$w\not\in\{\so^n,\sz^n\}$). Thus we have $f_w(k+1)=k$. Since $w\in\NPal(n)$ we 
have $|\Suff_k(w)|_{\so}=|\Pref_k(w)|_{\so}=k$ and 
$|\Suff_{k+1}(w)|_{\so}=|\Pref_{k+1}(w)|_{\so}=k$. This implies 
$f_{ww}(k+1)=k+1$ and hence $f_{ww}(k+1)\neq s_{ww}(k+1)$. The proof of $w\so 
w\not\in\NPal(2n+1)$ is similar to the proof of 
Lemma~\ref{ww}: in the middle of $w\so w$ is a larger block of $\so$'s than at 
the end.\qed
\end{proof}

\begin{lemma}\label{w0w}
Let $w\in\NPal(n)\backslash\{\sz^n\}$ with $n\in\N_{\geq 3}$. If $w\sz 
w$ 
is also a prefix normal palindrome then $w=\so^k$ or $w=\so^k\sz\so 
u\so\sz\so^k$ for some $u\in\Sigma^{\ast}$ and $k\in\N$.
\end{lemma}

\begin{proof}
Consider 
$w\in\NPal(n)\backslash\{\sz^n\}$. By $w\neq\sz^n$ follows $w=\so u\so$ for an 
$u\in\Pal(n)$. If $w\neq\so^k$, there exists $k$ minimal with $w[k]=\sz$. 
Suppose $w=\so^k\sz^{\ell}u\sz^{\ell}\so^k$ for some $k\in\N$ and 
$\ell\in\N_{>1}$. Then $p_{w\sz w}(k+2)=k$ but $|w\sz w[n-k,n+2]|_{\so}=k+1$. 
This is a contradiction to $w\sz w\in\NPal(2n+1)$. This implies $w=\so^k\sz\so 
u\so \sz\so^k$. \qed
\end{proof}

A characterisation for $w\so w$ being a \pnPal{} is more 
complicated. By $w\in\NPal(n)$ follows that a block of $\so$s contains at most 
the number of $\so$s of the previous block. But if such a block contains 
strictly less $\so$s the number of $\sz$s in between can increase by the same 
amount the number of $\so$s decreased.

\begin{lemma}\label{1ww1}
Let $w\in\NPal(n)\backslash\{\so^n,\sz^n\}$. If $\so ww\so$ is also a
prefix normal palindrome then $\so\sz\in\Pref(w)$.
\end{lemma}

\begin{proof}
Let $\so ww\so\in\NPal(2n+2)$.  Since $w\neq\sz^n$, there exists 
$u\in\Sigma^{\ast}$ with $w=\so u\so$. Since $w\neq\so^n$ there exists a 
minimal $k\in\N$ with $u[k]=\sz$. If $k>1$, then $\so 
ww\so=\so^k\sz\so^{2k}\sz\so^k$ or $\so ww\so = \so^k\sz v\sz\so^{2k}\sz 
v\sz\so^k$. In both cases a contradiction to $\so ww\so\in\NPal(2n+2)$. \qed
\end{proof}

Lemma~\ref{ww}, \ref{w0w}, and \ref{1ww1} indicate that a characterization of 
prefix normal palindromes based on smaller ones is hard to determine.

\section{Conclusion}

Based on the work in \cite{journals/corr/abs-1805-12405}, we investigated prefix normal palindromes in Section 3 and gave a characterisation based on the maximum-ones function. At the end of Section 4 results for a recursive approach to determine prefix normal palindromes are given. These results show that easy connections between prefix normal palindromes of different lengths cannot be expected. By introducing the collapsing relation we were able to partition the set of extension-critical words introduced in \cite{journals/corr/abs-1805-12405}. This leads to a characterization of collapsing words which can be extended to an algorithm determining the corresponding equivalence classes. Moreover we have shown that palindromes and the collapsing classes are related.

The concrete values for prefix normal palindromes and the index of the collapsing relation remain an open problem as well as the cardinality of the equivalence classes w.r.t. the collapsing relation.
Further investigations of the prefix normal palindromes and the collapsing classes lead directly to the index of the prefix equivalence.

\bigskip

\noindent\textbf{Acknowledgments.} We would like to thank Florin Manea for helpful discussions and advice.

\bibliographystyle{plainnat}
\bibliography{pnw}


\end{document}